\documentclass[runningheads]{llncs}
\usepackage{amssymb,amsmath}
\usepackage{mathtools}
\usepackage[hidelinks]{hyperref} 
\usepackage{xcolor}
\usepackage[capitalize,nameinlink,sort]{cleveref}
\usepackage{tikz}
\usetikzlibrary{automata, positioning, arrows, calc}
\usepackage{float}
\usepackage{subcaption}
\usepackage{graphicx}
\usepackage{ragged2e}

\spnewtheorem*{question*}{Question}{\bfseries}{}
\spnewtheorem{codes}{Code}{\bfseries}{}

\newcommand{\infw}[1]{\mathbf{#1}}
\DeclareMathOperator{\rep}{rep}
\DeclareMathOperator{\val}{val}

\makeatletter
\newcommand*{\rom}[1]{\expandafter\@slowromancap\romannumeral #1@}
\makeatother

\begin{document}
\title{Exploring the Crochemore and Ziv-Lempel factorizations of some automatic sequences with the software \texttt{Walnut}}
\titlerunning{The c-and z-factorizations of some automatic sequences via \texttt{Walnut}}
%
\author{Marieh Jahannia \inst{1}\orcidID{0000-0001-8510-2599} 
\and Manon Stipulanti\inst{2}\orcidID{0000-0002-2805-2465}}
\authorrunning{M. Jahannia and M. Stipulanti}
%
\institute{
School of Mathematics, Statistics and Computer Science, College of Science University of Tehran, Tehran, Iran \\
\email{mjahannia@ut.ac.ir}
\and 
 Department of Mathematics,	University of Li\`ege, Li\`ege, Belgium, \\\email{m.stipulanti@uliege.be}
 }
\maketitle              

\begin{abstract}
We explore the Ziv-Lempel and Crochemore factorizations of some classical automatic sequences making an extensive use of the theorem prover \verb|Walnut|.
\end{abstract}

\keywords{Combinatorics on words \and Crochemore factorization \and Ziv-Lempel factorization  \and Automatic sequences
\and \texttt{Walnut} theorem prover.}

\bigskip


\newpage

\section{Introduction}

In the expansive variety of tools at the heart of combinatorics on words, factorizations break down a given sequence into simpler components to provide valuable insights about its properties and behavior.
Crochemore~\cite{croch2,croch3} on the one hand and Lempel and Ziv~\cite{LZfact,Ziv} on the other introduced two such distinguished factorizations, taking after their respective names.
The first was aimed for algorithm design and builds on repetitive and non-repetitive aspects of sequences.
The second has remained a cornerstone of data compression and string processing algorithms, with ongoing discoveries revealing new applications for its use.
In the case of infinite words, Berstel and Savelli~\cite{Berstel06} characterized the Crochemore factorization of Sturmian words, the Thue-Morse sequence and its generalizations, and the period-doubling sequence.
Then Ghareghani et al.~\cite{ghareghani2020z} examined both the Crochemore and Ziv-Lempel factorizations for standard episturmian words.
Constantinescu and Ilie~\cite{CI2007} characterized the general behavior of the Ziv-Lempel factorization of morphic sequences, depending on the periodicity of the sequence and the growth function of the morphism.
Jahannia et al.~\cite{jahannia2018palindromic,jahannia2020closed} introduced two variations of these factorizations where the factors are required to satisfy an additional property and studied them for the $m$-bonacci words. 

Within combinatorics on words, the study of automatic sequences~\cite{Allouche,RM2002,Sha88} has turned out to be a fascinating journey, showing complex structures and patterns defined by fundamental mathematical principles. 
As their name points it out, these sequences are produced by finite automata with output, but also as fixed points of morphisms.
Their generation rules often lead to connections to various branches of mathematics.
The Thue-Morse sequence is one of the most famous --if not the most famous-- examples and encapsulates the binary encoding of the occurrences of $0$'s and $1$'s in binary representations, revealing a self-replicating and non-repetitive structure~\cite{Allouche}.
The period-doubling sequence, born from the logistic map, exhibits a chaotic behavior and contributes to the field of dynamical systems~\cite{Devaney}.
The Rudin-Shapiro sequence, recognized for its statistical relationship or mutual dependence to the Golay sequence, unfolds a binary pattern influenced by the existence of certain arithmetic progressions~\cite{brillhart1991case}.
Lastly, the paper-folding sequence, a classic representation of a fractal, illustrates how simple operations can generate infinitely long sequences~\cite{Allouche}.

Recent developments in the field of sequence analysis have seen the emergence of systematic and automated decision procedures designed to autonomously determine the validity of a given property for specified sequences, altering the need of human work in proofs.
Notably, in the case of automatic sequences, Mousavi~\cite{mousavi2021automatic} and Shallit~\cite{Shallit2022logical} have made significant contributions by developing the software called \verb|Walnut|.
The latter works by representing a sequence as a finite automaton and expressing properties as first-order logic predicates. The decision procedure then translates these predicates into automata, facilitating the identification of representations for which the predicate holds true.

In this paper, we make use of \verb|Walnut| to obtain a precise description of the Crochemore and Ziv-Lempel factorizations of some classical automatic sequences.
More precisely, our main direction of investigation is the following general question (see Section~\ref{sec:background} for precise definitions):
\begin{question*}
Given an abstract numeration system $S$ and an $S$-automatic sequence $\infw{x}$, is it possible to use \verb|Walnut| to show that the starting positions and lengths of the factors in both the Crochemore and Ziv-Lempel factorizations of $\infw{x}$ only depend on the numeration system $S$?
\end{question*}

Using Fici's~\cite{fici2015factorizations} nice survey of factorizations of the Fibonacci word, we produce a detailed \verb|Walnut| code in Section~\ref{sec:Fibo} to answer this question about the Fibonacci word.
In Section~\ref{sec: z and c for classical aut seq}, we apply this code to other classical automatic sequences.
We end the paper with Section~\ref{sec:conclusion} where we discuss the scope of our method.


\section{Background}
\label{sec:background}

\textbf{Combinatorics on words.}
We let $\Sigma$ denote a finite set of symbols, called the \emph{letters}, referred to as an \emph{alphabet}.
A \emph{word} over $\Sigma$ is a finite or infinite sequence of letters chosen from $\Sigma$.
In this paper, to differentiate finite and infinite words, we write the latter in bold.
As usual, we let $\Sigma^{*}$ represent the set of finite words over $\Sigma$, and $\varepsilon$ denote the \emph{empty} word.
For a word $w \in \Sigma^*$, we let $|w|$ denote its length. For all $i\in\{0,\ldots,|w|-1\}$, we use $w[i]$ to refer to the $i$-th letter of $w$, starting from position $0$. 
If we write $w = w_0w_1\cdots w_{|w|-1}$, then we let $\widetilde{w}$ denote its \emph{reversal} or \emph{mirror}, defined as  $\widetilde{w}=w_{|w|-1}\cdots w_1 w_0$.
A \emph{factor} of $w$ is a contiguous block of letters within $w$; we write $w[i..j]$ to represent the factor occupying positions $i, i+1, \ldots, j$.
For instance, if $w = \texttt{computer}$, then $v = \texttt{comp}$ is a factor of $w$ where $v = w[0..3]$.
A \emph{prefix} (resp., \emph{suffix}) of $w$ is a word $x$  such that $w = xy$ (resp., $w = yx$) for some word $y$. For instance, \texttt{comp} is a prefix and \texttt{er} is a suffix of \texttt{computer}.
If $w=xy$, we write $x^{-1}w= y$ and $wy^{-1} = x$.

A \emph{morphism} is a mapping $\psi \colon \Sigma^{*} \rightarrow \Sigma^{*}$ such that for all $u, v \in \Sigma^{*}$, $\psi(uv) = \psi(u)\psi(v)$.
A morphism $\psi$ is \emph{$k$-uniform} if there exists an integer $k$ such that $|\psi(a)|=k$ for all $a \in \Sigma$.
A \emph{coding} is a $1$-uniform morphism.
The morphism $\psi$ is \emph{prolongable} on the letter $a\in\Sigma$ if $\psi(a) = au$ and $\psi^{n}(u) \neq \varepsilon$ for all $n \geq 0$.
A \emph{fixed point} of $\psi$ is given by $\psi^{\omega}(a) = au \psi(u) \psi^{2}(u) \cdots $.
For example, the morphism $\mu \colon a \mapsto ab, 1 \mapsto ba$ is prolongable on both letters $a$ and $b$. The infinite word $t =\mu^{\omega}(a)= abba baab baab abba baab abba \cdots$ is a fixed point called the \emph{Thue-Morse} sequence.

\textbf{Factorizations.}
For a finite word $w$, a \emph{factorization} of $w$ is a sequence $(x_0, x_1, \ldots,x_m)$ of finite words such that $w$ can be expressed as the concatenation of the elements of the sequence, i.e., $w=x_0x_1 \cdots x_m $. Similarly, in the case of an infinite word $\infw{w}$, the factorization is a sequence $(x_0,x_1,\ldots)$ of finite words such that $\infw{w}= x_0x_1\cdots$.
For example, a factorization of the word $w = \texttt{abracadabra}$ is $(\texttt{ab},\texttt{ra}, \texttt{ca}, \texttt{da}, \texttt{bra})$.
We now introduce two distinguished factorizations of words.
More formally, given an infinite word $\infw{w}$, its \emph{Ziv-Lempel factorization}, in short \emph{$z$-factorization}, is given by $z(\infw{w}) = (z_0, z_1, \ldots)$, where $z_m$ is the shortest prefix of $z_mz_{m+1}\cdots$ occurring only once in the word $z_0 \cdots z_m$.
The \emph{Crochemore factorization}, in short \emph{$c$-factorization}, of $\infw{w}$ is given by $c(\infw{w}) = (c_0, c_1, \ldots)$, where $c_m$ is either the longest prefix of $c_mc_{m+1} \cdots$ occurring twice in $c_0 \cdots c_m$ or a letter not present in $c_0 \cdots c_{m-1}$.
Roughly, the Ziv-Lempel factorization breaks the sequence into minimal never seen before factors, while the Crochemore one splits the sequence into maximal already seen factors.
We consider similar definitions on finite words.
For instance, for $w = abbabaabbaababb$, its $z$- and $c$-factorizations are respectively $z(w) = (a, b, ba, baa, bbaa, babb)$ and $c(w) = (a,b,b,ab,a,abba,aba,bb)$.

\textbf{Abstract numeration systems.}
Such numeration systems were introduced at the beginning of the century by Lecomte and Rigo~\cite{Lecomte-Rigo-2001}; see also~\cite[Chap.~3]{CANT10} for a general presentation.
An abstract numeration system (ANS) is defined by a triple $S=(L,\Sigma,<)$ where $\Sigma$ is an alphabet ordered by the total order $<$ and $L$ is an infinite regular language over $\Sigma$, i.e., accepted by a deterministic finite automaton. 
We say that $L$ is the \emph{numeration language} of $S$.
When we genealogically order the words of $L$, we obtain a one-to-one correspondence $\rep_S$ between $\mathbb{N}$ and $L$.
Then, the \emph{$S$-representation} of the non-negative integer $n$ is the $(n+1)$st word of $L$, and the inverse map, called the \emph{(e)valuation map}, is denoted by $\val_S$.
For instance, consider the ANS $S$ built on the language $a^* b^*$ over the ordered alphabet $\{a,b : a<b\}$. 
  The first few words in the language are $\varepsilon,a,b,aa,ab,bb,aaa$, and we have $\rep_S(5)=bb$ and $\val_S(aaa)=6$.

\textbf{Automatic sequences.}
As their name indicates it, automatic sequences are defined through automata.
A \emph{deterministic finite automaton with output (DFAO)} is defined by a 6-tuple $\mathcal{M}=(Q, \Sigma, \delta, q_0,\Delta, \tau)$, where $Q$ is a finite set of \emph{states}, $\Sigma$ is a finite \emph{input alphabet}, $\delta\colon Q \times \Sigma \rightarrow Q$ is the \emph{transition function}, $q_0$ is the \emph{initial state}, $\Delta$  is a finite \emph{output alphabet}, and $\tau\colon Q \rightarrow \Delta$ is the \emph{output function}.
The output of $M$ on the finite word $w \in \Sigma^{*}$, denoted $M(w)$, is defined as $M(w) = \tau (\delta(q_0,w))\in \Delta$.

Let $S=(L,\Sigma,<)$ be an ANS.
An infinite word $\mathbf{x}$ is {\em $S$-automatic} if there exists a DFAO $\mathcal{M}$ such that, for all $n\ge 0$, the $n$th term $\mathbf{x}[n]$ of $\infw{x}$ is given by the output $\mathcal{M}(\rep_S(n))$ of $\mathcal{M}$.
In this case, we say that the DFAO $\mathcal{M}$ \emph{generates} or \emph{produces} the sequence  $\infw{x}$.
In particular, for an integer $k\ge 2$, if $\mathcal{M}$ is fed with the genealogically ordered language $\{\varepsilon\} \cup \{1,\ldots,k-1\}\{0,\ldots,k-1\}^*$, then $\mathbf{x}$ is said to be {\em $k$-automatic}.
For the case of integer base numeration systems, a classical reference on automatic sequences is~\cite{Allouche}, while~\cite{RM2002,Sha88} treat the case of more exotic numeration systems.
A well-known characterization of automatic sequences states that they are morphic, i.e., obtained as the image under a coding of a fixed point of a morphism~\cite{RM2002}.
In particular, a sequence is $k$-automatic if and only if the morphism producing it is $k$-uniform morphism~\cite{Allouche}.

\begin{figure}
\centering
\begin{subfigure}{0.25\textwidth}
    \begin{tikzpicture}[>=stealth,node distance=2cm,on grid,auto,scale=0.75, every node/.style={scale=0.75}]
    \node[state, initial,initial text=] (q0) {$0/a$};
    \node[state, right=of q0] (q1) {$1/b$};
    \path[->]
        (q0) edge [loop above] node {0} ()
             edge [bend left] node {1} (q1)
        (q1) edge [bend left] node {0} (q0);
    \end{tikzpicture}
    \caption{Fibonacci.}
    \label{fig:FIB:DFAO}
\end{subfigure}
\hfill
\begin{subfigure}{0.25\textwidth}
    \begin{tikzpicture}[>=stealth,node distance=2cm,on grid,auto,scale=0.75, every node/.style={scale=0.75}]
    \node[state, initial, initial text=] (q0) {$q_0/a$};
    \node[state, right=of q0] (q1) {$q_1/b$};
    \path[->]
        (q0) edge [loop above] node {$0$} ()
             edge [bend left] node {$1$} (q1)
        (q1) edge [loop above] node {$0$} () 
             edge [bend left] node {$1$} (q0);
    \end{tikzpicture}
    \caption{Thue-Morse.}
    \label{fig:TM:DFAO}
\end{subfigure}
\hfill
\begin{subfigure}{0.25\textwidth}
     \begin{tikzpicture}[->,>=stealth',shorten >=1pt,auto,node distance=2cm,scale=0.75, every node/.style={scale=0.75}]
  \tikzstyle{every state}=[fill=white,text=black]
  
  \node[state,initial, initial text=] (0) {$q_0/a$};
  \node[state] (1) [right of=0] {$q_1/b$};
  
  \path (0) edge [loop above] node {$0$} (0)
        (0) edge [bend left] node {$1$} (1)
        (1) edge [bend left] node {$0,1$} (0);
  \end{tikzpicture}
  \caption{Period-doubling.}
  \label{fig:PD:DFAO}
\end{subfigure}
\\
\begin{subfigure}{0.25\textwidth}
    \begin{tikzpicture}[>=stealth,node distance=2cm,on grid,auto,scale=0.75, every node/.style={scale=0.75}]
		\node[state, initial, initial text=] (q0) {$q_0/a$};
		\node[state, right=of q0] (q1) {$q_1/b$};
		\path[->]
		(q0) edge [loop above] node {$0,1$} ()
		edge [bend left] node {$2$} (q1)
		(q1) edge [loop above] node {$0,1$} () 
		edge [bend left] node {$2$} (q0);
	\end{tikzpicture}
	\caption{Mephisto-Waltz.}
	\label{fig:M:DFAO}
\end{subfigure}
\hfill
\begin{subfigure}{0.7\textwidth}
    \begin{tikzpicture}[shorten >=0.25pt,node distance=2cm,on grid,auto,scale=0.75, every node/.style={scale=0.75}] 
   \node[state,initial, initial text=] (q0)   {$q_0/+1$}; 
   \node[state] (q1) [right=of q0] {$q_1/+1$}; 
   \node[state] (q2) [right=of q1] {$q_2/-1$}; 
   \node[state] (q3) [right=of q2] {$q_3/-1$};
    \path[->] 
    (q0) edge [bend left] node {1} (q1)
    (q1) edge [bend left] node {1} (q2)
    (q2) edge [bend left] node {1} (q1)
    (q2) edge [bend left] node {0} (q3)
    (q3) edge [loop above] node {0} ()
    (q3) edge [bend left] node {1} (q2)
    (q1) edge  [bend left] node {0} (q0);
    
    \draw (q0) edge [loop above] node {0} ();
    \end{tikzpicture}
    \caption{Rudin-Shapiro.}
    \label{fig:DFAO:RS}
\end{subfigure}
\\
\begin{subfigure}{0.9\textwidth}
\begin{tikzpicture}[shorten >=0.75pt,node distance=2cm,on grid,auto,scale=0.75, every node/.style={scale=0.75}] 
   \node[state,initial, initial text=] (q0)   {$q_0/+1$}; 
   \node[state] (q1) [right=of q0] {$q_1/+1$}; 
   \node[state] (q2) [right=of q1] {$q_2/-1$}; 
   \node[state] (q3) [right=of q2] {$q_3/-1$};
    \path[->] 
    (q0) edge [bend left] node {1} (q1)
    (q1) edge [bend left] node {0} (q2)
    (q2) edge [bend left] node {1} (q3)
    (q2) edge [bend left] node {0} (q0)
    (q3) edge [loop above] node {1} ()
    (q1) edge [loop above] node {1} ()
    (q3) edge [bend left] node {0} (q2)  ;  
    \draw (q0) edge [loop above] node {0} ();
    \end{tikzpicture}
    \caption{Paper-folding.}
    \label{fig:DFAO:PF}
    
\end{subfigure}
        
\caption{DFAOs generating the automatic sequences of the paper.}
\label{fig:figures}
\end{figure}
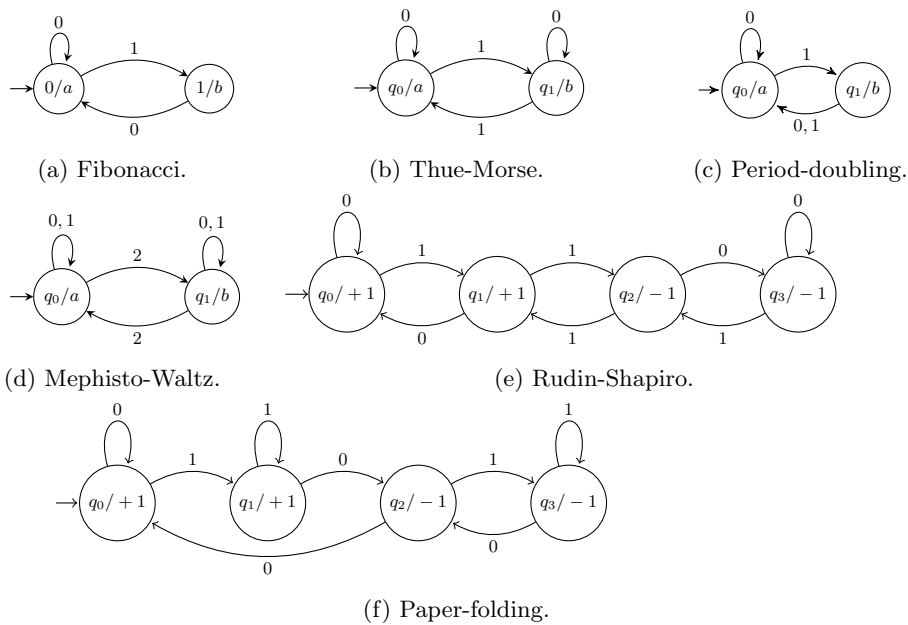

\section{The Fibonacci word}
\label{sec:Fibo}

The infinite Fibonacci word $\infw{f}= abaababa \ldots $  is the fixed point, starting with $a$, of the morphism $\phi \colon a \rightarrow ab, b \rightarrow a$.
It is automatic in a specific ANS based on Fibonacci numbers defined by $F_0 = 1$, $F_1 = 2$, and $F_n = F_{n-1} + F_{n-2}$ for all $n\ge 2$.
Zeckendorf~\cite{Zekendorf} demonstrated a remarkable theorem stating that every non-negative integer can be represented as a sum of distinct non-consecutive Fibonacci numbers.
Given an integer $n$, its \textit{canonical Fibonacci representation}, denoted as  $\rep_F(n)$, is defined as $\rep_F(n) = \sum_{i=0}^{m} c_i F_i$ where the coefficients $c_i$ are in $\{0,1\}$ and obtained using the greedy algorithm.
This gives rise to the \emph{Fibonacci} or \emph{Zeckendorf numeration system} and the corresponding \emph{Fibonacci-automatic} sequences.
For instance, the automaton in Figure~\ref{fig:FIB:DFAO} generates $\infw{f}$.

Related to our concern, Fici~\cite{fici2015factorizations} showed that the $z$-factorization of the Fibonacci word $\infw{f}$ is the concatenation of its singular words, i.e., 
\begin{align}
\label{Equ:FibZFactFici}
	\infw{f}=\prod_{n \geq -1} w_{n} 
                =a  \cdot b \cdot aa \cdot bab \cdot aabaa \cdot babaabab \cdots.
\end{align}
The \emph{singular words} of the Fibonacci word $\infw{f}$, introduced by Wen and Wen~\cite{WenWen}, are defined as follows: $w_{-1}=a$,  $w_0 = b$, and for $n \geq 1$,  $w_n = x \cdot \phi^n(a) \cdot y^{-1} $, where  $xy \in \{ab, ba\}$ is the length-$2$ suffix of $\phi^n(a)$. 
The first few singular words are $a$, $b$, $aa$,  $bab$, $aabaa$, $babaabab$, $aabaababaabaa$.
Notably,  $|w_n| = |\phi^n(a)| = F_n$  for all  $n \geq 0$.

The idea behind our approach is that we can express the $z$-factorization of $\infw{f}$ as a first-order logic formula in \verb|Walnut|.
We define the two following predicates:
\begin{verbatim}
def fibfactoreq "?msd_fib At t<n => F[i+t]=F[j+t]":
def fibzfactor "?msd_fib (Aj j<i => ~$fibfactoreq(i,j,n)) 
    & (At t<n => (El l<i => $fibfactoreq(i,l,t)))":        
\end{verbatim}
The first formula takes the triple $(i,j,n)$ as input and checks whether the length-$n$ factors $\infw{f}[i..i+n-1]$ and $\infw{f}[j..j+n-1]$ are equal.
The second formula, on input $(i,n)$, verifies whether the factor $\infw{f}[i..i+n-1]$ does not appear before position $i$ and each of its prefixes appears before.
In other words, it tests whether $\infw{f}[i..i+n-1]$ is the shortest prefix of $(\infw{f}[0...i-1])^{-1}\infw{f}$ occurring only once in $\infw{f}[0...i+n-1]$.
In particular, the variables $i,j$ indicate positions within $\infw{f}$ and $n$ is a measure of length.
Running \verb|Walnut| on these predicates yields the automaton in Figure~\ref{fig:Walnut:FibZFact}.
\begin{figure}
    \centering
    \includegraphics[scale=0.45]{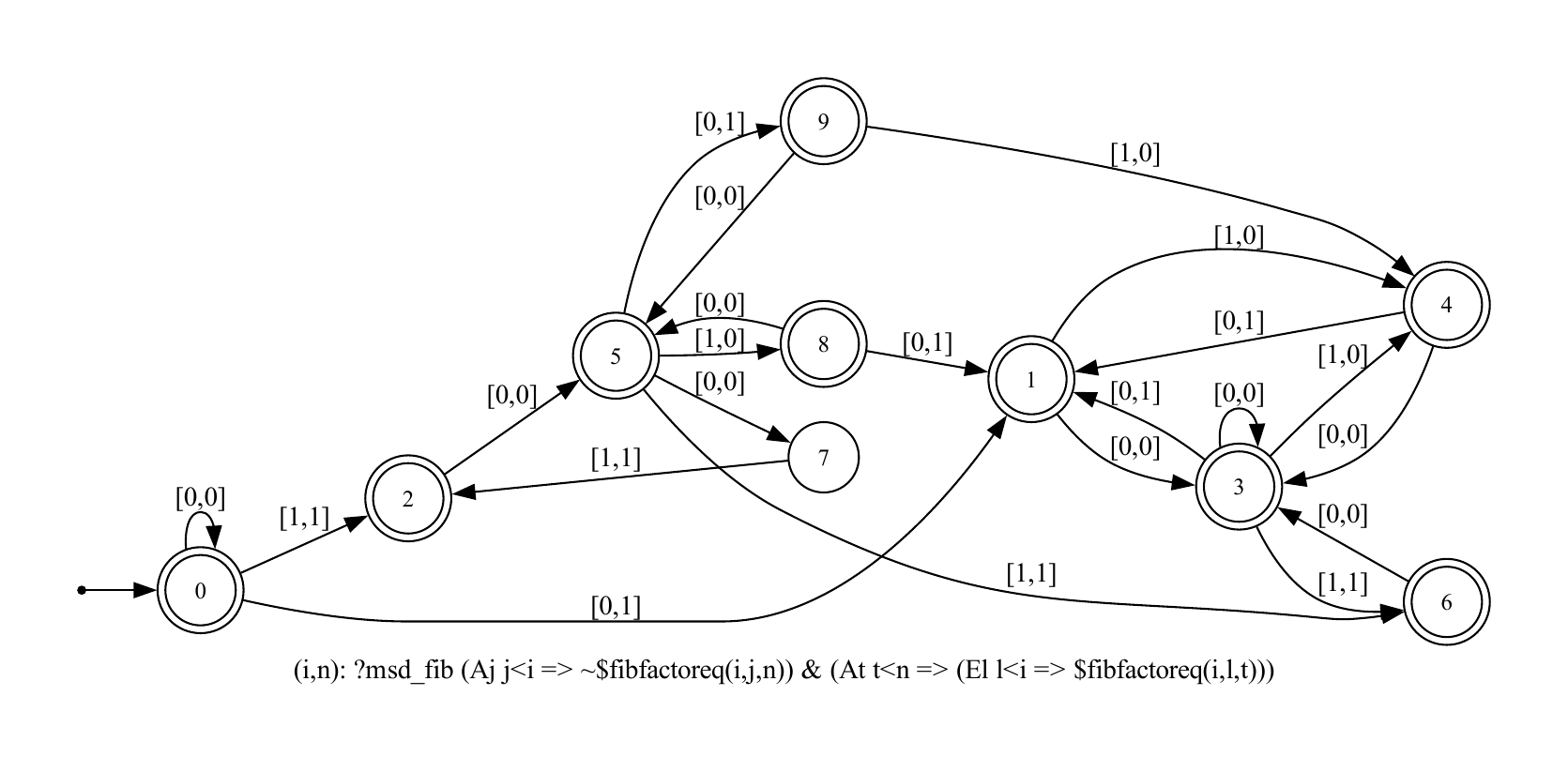}
    \caption{An automaton accepting, among others, the base-$2$ representations of the pairs $(i,n)$ giving the position and length of factors of the $z$-factorization of the Fibonacci word.}
    \label{fig:Walnut:FibZFact}
\end{figure}
Now observe from Identity~\eqref{Equ:FibZFactFici} that the pairs  $(i, n)$ of position and length of the factors of the $z$-factorization of $\infw{f}$ are given by $(0,1)$ and $(F_{n+1}-1, F_n)$ for all $n\ge 0$. 
This gives rise to the regular expression in \verb|Walnut|:
\begin{verbatim}
reg fibzgoodrep msd_fib msd_fib "[0,0]*[0,1] | [0,0]*[1,1] 
    | [0,0]*[1,1][0,0]([1,0][0,0])* 
    | [0,0]*[1,1][0,0]([1,0][0,0])*[1,0]":
\end{verbatim}
Now we check that the pairs $(i,n)$ guessed before indeed give factors having the desired property \verb|fibzfactor|, so we write the following check in \verb|Walnut|:
\begin{verbatim}
eval fibzcheck "?msd_fib Ai An $fibzgoodrep(i,n) => $fibzfactor(i,n) ":
\end{verbatim}
and \verb|Walnut| returns \verb|TRUE|.
Note that previous check is not a bi-implication.
Indeed, for instance, we see from Figure~\ref{fig:Walnut:FibZFact} that other pairs $(i,n)$ satisfy \verb|fibzfactor|.
This is for instance the case of $(3,4)$ since $\infw{f}[3..6]=abab$ is the shortest prefix of $(\infw{f}[0...2])^{-1}\infw{f}=ababaabaababaaba\cdots$ that occurs only once in $\infw{f}[0...6]=abaabab$.
Therefore, to obtain the $z$-factorization of $\infw{f}$, we need the final check \verb|fibcheck| to be true, but also to have consecutive positions $i$ that cover all $\mathbb{N}$.

Fici~\cite{fici2015factorizations} showed the following other factorization of the Fibonacci word:
\begin{align}
\label{Equ:FibCFactFici}
    \infw{f} = \prod_{n \geq 1} \widetilde{\phi^n(a)}
            =a \cdot ba \cdot aba \cdot baaba \cdot ababaaba \cdots.
\end{align}
In fact, the latter is almost the $c$-factorization of $\infw{f}$, with the difference that
the $c$-factorization starts with $a$, $b$, $a$ and then coincides with Identity~\eqref{Equ:FibCFactFici}.
See~\cite{Berstel06}.
It is not difficult to modify the formula \verb|fibzfactor| in \verb|Walnut| to deal with the $c$-factorization instead: 
\begin{verbatim}
def fibcfactor "?msd_fib (Ej j<i => $fibfactoreq(i,j,n)) 
        & (Al l<i => ~$fibfactoreq(i,l,n+1))":
\end{verbatim}
As before, examining Identity~\eqref{Equ:FibCFactFici}, the pairs $(i, n)$ for the factors of the $c$-factorization are given by the following regular expression in \verb|Walnut|:
\begin{verbatim}
reg fibcgoodrep msd_fib msd_fib "[0,0]*[0,1] | [0,0]*[1,1] 
    | [0,0]*[1,0][0,1] | [0,0]*[1,1][0,0][0,0] 
    | [0,0]*[1,1][0,0][0,0][1,0] 
    | [0,0]*[1,1][0,0]([1,0][0,0])*[1,0][0,0][0,0] 
    | [0,0]*[1,1][0,0]([1,0][0,0])*[1,0][0,0][0,0][1,0]":
    \end{verbatim}
The last check 
\begin{verbatim}
eval fibccheck "?msd_fib Ai An $fibcgoodrep(i,n) => $fibcfactor(i,n)":
\end{verbatim}
returns \verb|TRUE|.

\begin{remark}
 The palindromic (resp., closed) version of the $z$- and $c$-factorizations, defined in~\cite{jahannia2018palindromic} (resp.,~\cite{jahannia2020closed}), requires that each factor is palindromic (resp., closed, i.e., each factor has a proper factor that occurs exactly twice, as a prefix and as a suffix).
 One can tweak the \verb|Walnut| code presented in this section to obtain these.
 See~\cite[Sec.~8.6.3 and 8.8.3]{Shallit2022logical} for related \verb|Walnut| code.
\end{remark}

\begin{table}[]
    \centering
    \[
    \begin{array}{c|l}
        \infw{x} & z(\infw{x})=(z_0,z_1,\ldots) \\
         \hline
         \infw{f} & (a  , b , aa , bab , aabaa , babaabab , aabaababaabaa, 
 babaababaabaababaabab,\ldots ) \\
         \infw{t} & (a , b , ba , baa , bbaa , babb  , abaaba , bbaabb, abaabbaababbaa, bbabaababbab,\ldots) \\
         \infw{pd} & (a , b , aa , abab , abaaabaa , abaaabababaaabab, abaaabababaaabaaabaaabababaaabaa, \ldots ) \\
         \infw{rs} & (1 , 11(-1), 11(-1)111, 1(-1)(-1), (-1)1(-1) , 111(-1)11(-1)1(-1) , (-1)(-1)11,\ldots ) \\
         \infw{pf} & (1,1(-1),11(-1)(-1),111,(-1)(-1)1(-1),(-1)111(-1)1, 1(-1)(-1)(-1), \ldots) \\
         \infw{mw} & (a , ab , aabb , baaa , baabbbab  , babbaa, abbbaaabaabbbabbabbaaabb,\ldots ) \\
         \hline
         \hline
         \infw{x} & c(\infw{x})=(c_0,c_1,\ldots) \\
         \hline
         \infw{f} & (a , b , a , aba , baaba , ababaaba , baabaababaaba , ababaababaabaababaaba,\ldots ) \\
         \infw{t} & (a , b , b , ab , a , abba , aba , bbabaab , abbaab, babaabbaababba, abbabaababba,\ldots ) \\
         \infw{pd} & (a , b , a,aa, ba,baba,aaba, aabaaaba, babaaaba, babaaabababaaaba, \ldots ) \\
        \infw{rs} & (1, 11  , (-1) , 11(-1)11 , 11(-1) , (-1)(-1) , 1(-1)111 , (-1)11(-1)1 , (-1)(-1)(-1)1,\ldots ) \\
         \infw{pf} & (1,1,(-1),11(-1),(-1)11,1(-1)(-1)1,(-1)(-1)111(-1), 11(-1)(-1),\ldots) \\
         \infw{mw} & (a , a , b , aab , bb , aa , abaabbba , bbabba, aabaab, aabbbaaabaabbbabbabbaaab, \ldots )
    \end{array}
    \]
    \caption{The first few factors of the $z$- and $c$-factorizations of the automatic sequences considered in this paper.}
    \label{tab:z and c fact of our sequences}
\end{table}

\section{The $z$- and $c$- factorizations of some classical automatic sequences}
\label{sec: z and c for classical aut seq}

The conclusion of the previous section is the following: when a candidate is known for the $z$- or $c$-factorization of an infinite word, then it is not difficult to check with \verb|Walnut| that it is indeed the right factorization.
This is the purpose of the current section, and we use the same techniques as in Section~\ref{sec:Fibo}.
Given an infinite word $\infw{x}$, the \verb|Walnut| code for its $z$-factorization (resp., $c$-factorization) is summed up in Code~\ref{code:z-fact} (resp., Code~\ref{code:c-fact}).
Also see Table~\ref{tab:z and c fact of our sequences} where we give the first few factors of $z(\infw{x})$ and $c(\infw{x})$ for some words $\infw{x}$.

\begin{codes}
\label{code:z-fact}
Given the $k$-automatic sequence $\infw{x}$ coded by \verb|X| in \verb|Walnut|, we use the following predicates to find its $z$-factorization, where \verb|LX| is some specific guessed regular expression:
\begin{verbatim}
def xfactoreq "?msd_k At t<n => X[i+t]=X[j+t]":
def xzfactor "?msd_k (Aj j<i => ~$xfactoreq(i,j,n)) & 
    (At t<n => (El l<i => $xfactoreq(i,l,t)))":
reg xzgoodrep msd_k msd_k "LX":
eval xzcheck "?msd_k Ai An $xzgoodrep(i,n) => $xzfactor(i,n)":
\end{verbatim}
\end{codes}

\begin{codes}
\label{code:c-fact}
Given the $k$-automatic sequence $\infw{x}$ coded by \verb|X| in \verb|Walnut|, we use the following predicates to find its $c$-factorization, where \verb|LX| is some specific guessed regular expression:
\begin{verbatim}
def xfactoreq "?msd_k At t<n => X[i+t]=X[j+t]":
def xcfactor "?msd_k (Ej j<i => $xfactoreq(i,j,n)) 
    & (Al l<i => ~$xfactoreq(i,l,n+1))":
reg xcgoodcrep msd_k msd_k "LX":
eval xccheck "?msd_k Ai An $xcgoodrep(i,n) => $xcfactor(i,n)":
\end{verbatim}
\end{codes}

\subsection{The Thue-Morse sequence}

The most famous example among $2$-automatic sequences is the \emph{Thue-Morse} sequence $\infw{t}$ which is the fixed point of the morphism $\mu \colon a \mapsto ab, b \mapsto ba$  starting with $a$. This sequence is generated by the automaton in Figure~\ref{fig:TM:DFAO}.

\begin{theorem}
    Let $z(\infw{t})=(z_0,z_1,\ldots)$ be the $z$-factorization of the Thue-Morse sequence $\infw{t}$.
    Then, for all $m\in\{0,\ldots,6\}$, $z_m$ is given in Table~\ref{tab:z and c fact of our sequences} and, for all $m\ge 7$, $z_m=\infw{t}[i..i+n-1]$ where
    \[
    (i,n)
    = 
     \begin{cases}
    (13\cdot 2^{m/2-3}+1, 7\cdot 2^{m/2-3}), &\text{if $m$ is even;} \\
    (5\cdot 2^{(m-1)/2-1}+1, 3\cdot 2^{(m-1)/2-2}), &\text{if $m$ is odd.}
    \end{cases}
    \]
\end{theorem}
\begin{proof}
In Code~\ref{code:z-fact}, replace \verb|X| by \verb|T|, \verb|k| by $2$, and \verb|LX| by
\begin{verbatim}
[0,0]*[0,1] | [0,0]*[1,1] | [0,0]*[1,0][1,0] 
| [0,0]*[1,0][0,1][0,1] | [0,0]*[1,0][1,0][1,0] 
| [0,0]*[1,0][0,1][1,0][1,0] | [0,0]*[1,0][1,0][1,1][1,0] 
| [0,0]*[1,0][1,1][0,1][1,1][0,0]*[1,0] 
| [0,0]*[1,0][0,0][1,1][0,1][0,0]*[1,0]
\end{verbatim}
Then running Code~\ref{code:z-fact} in \verb|Walnut| returns \verb|TRUE|.
\qed
\end{proof}

The next result gives back~\cite[Thm.~2]{Berstel06}.

\begin{theorem}
Let $c(\infw{t})=(c_0,c_1,\ldots)$ be the $c$-factorization of the Thue-Morse sequence $\infw{t}$.
Then, for all $m\in\{0,\ldots,5\}$, $c_m$ is given in Table~\ref{tab:z and c fact of our sequences} and, for all $m\ge 6$, $c_m = \infw{t}[i..i+n-1]$ where
 \[
    (i,n)
    = 
     \begin{cases}
        (5\cdot 2^{m/2-2}, 3\cdot 2^{m/2-3}), &\text{if $m$ is even;} \\
        (13\cdot 2^{(m-1)/2-3}, 7\cdot 2^{(m-1)/2-3}), &\text{if $m$ is odd.}
    \end{cases}
    \]
\end{theorem}
\begin{proof}
In Code~\ref{code:c-fact}, replace \verb|X| by \verb|T|, \verb|k| by $2$, and \verb|LX| by
\begin{verbatim}
[0,0]*[0,1] | [0,0]*[1,1] | [0,0]*[1,0][0,1] | [0,0]*[1,1][1,0] 
| [0,0]*[1,0][0,0][1,1] | [0,0]*[1,1][1,0][0,0]  
| [0,0]*[1,0][0,0][1,1][0,1][0,0]* | [0,0]*[1,0][1,1][0,1][1,1][0,0]*
\end{verbatim}
Then running Code~\ref{code:c-fact} in \verb|Walnut| returns \verb|TRUE|.
\qed
\end{proof}

\subsection{The period-doubling sequence}

The \emph{period-doubling} sequence $ \infw{pd} = ab aa ab ab ab aa ab aa \cdots $ is also a $2$-automatic sequence. It is closely related to the Thue-Morse sequence as it is defined for all $n\ge 0$ by $\infw{pd}[n]=b$ if $\infw{t}[n] = \infw{t}[n+1]$, $\infw{pd}[n]=a$ otherwise.
Furthermore, $ \infw{pd}$  is the fixed point of the morphism $h \colon a \mapsto ab , b \mapsto aa$, starting with $a$ and is generated by the automaton in Figure~\ref{fig:PD:DFAO}.

\begin{theorem}
Let $z(\infw{pd})=(z_0,z_1,\ldots)$ be the $z$-factorization of the period-doubling sequence $\infw{pd}$.
Then, $z_0=a$ and, for all $m\ge 1$, $z_m=\infw{pd}[i..i+n-1]$ where $i=n=2^{m-1}$.
\end{theorem}
\begin{proof}
In Code~\ref{code:z-fact}, replace \verb|X| by \verb|PD|, \verb|k| by $2$, and \verb|LX| by
\begin{verbatim}
[0,0]*[0,1] | [0,0]*[1,1][0,0]*
\end{verbatim}
Then running Code~\ref{code:z-fact} in \verb|Walnut| returns \verb|TRUE|.
\qed
\end{proof}

The next result gives back~\cite[Thm.~4]{Berstel06}.

\begin{theorem}
Let $c(\infw{pd})=(c_0,c_1,\ldots)$ be the $c$-factorization of the period-doubling sequence $\infw{pd}$.
Then $c_0=a$ and, for all $m\ge 1$, $c_m = \infw{pd}[i..i+n-1]$ where
\[
   (i,n)
    = 
     \begin{cases} 
        (3\cdot 2^{m/2-1}-1, 2^{m/2-1}), &\text{if $m$ is even;} \\
        (2^{(m-1)/2+1}-1, 2^{(m-1)/2}), &\text{if $m$ is odd.}
    \end{cases}
    \]
\end{theorem}
\begin{proof}
In Code~\ref{code:c-fact}, replace \verb|X| by \verb|PD|, \verb|k| by $2$, and \verb|LX| by
\begin{verbatim}
[0,0]*[0,1] | [0,0]*[1,1][1,0]* | [0,0]*[1,1][0,0][1,0]*
\end{verbatim}
Then running Code~\ref{code:c-fact} in \verb|Walnut| returns \verb|TRUE|.
\qed
\end{proof}

\subsection{The Rudin-Shapiro sequence}

The \emph{Rudin-Shapiro} sequence $\infw{rs}=111(-1)11(-1)1\cdots$ is defined as follows: for all $n\ge 0$, the $n$th letter $\infw{rs}[n]$ is given by $1$ or $-1$ according to the parity of the number of (possibly overlapping) occurrences of the block $11$ is the base-$2$ representation of $n$.
The sequence is $2$-automatic as the automaton in Figure~\ref{fig:DFAO:RS} generates it.
In addition, $\infw{rs}$ can be written as $\tau(\rho^\omega(a))$, where $\rho \colon a \mapsto ab, b \mapsto ac, c \mapsto db, d \mapsto dc$ and $\tau \colon a,b\mapsto 1, c,d\mapsto -1$.

\begin{theorem}
    Let $z(\infw{rs})=(z_0,z_1,\ldots)$ be the $z$-factorization of the Rudin-Shapiro sequence $\infw{rs}$.
    Then $z_m=\infw{rs}[i..i+n-1]$ where, for $0\le m \le 10$, $(i,n)$ belongs to 
    \[
    \{
    (0,1), (1,3), (4,6), (10,3), (13,3), (16,9), (25,4), (29, 8), (37, 12), (49, 6), (55,6)
    \},
    \]
    and for all $m\ge 11$ with $p = \lfloor \frac{m}{4} \rfloor$,
    \[
    (i,n)
    = 
	\begin{cases}
		(9\cdot 2^{p} + 1, 3\cdot 2^{p}), &\text{if $m\equiv 0 \bmod{4}$;} \\
		(3\cdot 2^{p+2} +1 ,  2^{p}), &\text{if $m\equiv 1 \bmod{4}$;} \\
		(13\cdot 2^{p} + 1 ,  2^{p+1}), &\text{if $m\equiv 2 \bmod{4}$;}\\
        (15\cdot 2^{p} + 1 , 3\cdot 2^{p}), &\text{if $m\equiv 3 \bmod{4}$.}
        \end{cases}
    \]
\end{theorem}
\begin{proof}
In Code~\ref{code:z-fact}, replace \verb|X| by \verb|RS|, \verb|k| by $2$, and \verb|LX| by
\begin{verbatim}
[0,0]*[0,1] | [0,0]*[1,0][1,1] | [0,0]*[1,1][0,1][0,0]
| [0,0]*[1,0][0,0][1,1][0,1] | [0,0]*[1,0][1,0][1,1][1,1]
| [0,0]*[1,0][0,1][0,0][0,0][0,1] | [0,0]*[1,0][1,0][0,1][0,0][1,0]
| [0,0]*[1,0][1,1][1,0][0,0][1,0] | [0,0]*[1,0][0,0][0,1][1,1][0,0][1,0]
| [0,0]*[1,0][1,0][0,0][0,1][0,1][1,0]
| [0,0]*[1,0][1,0][0,0][1,1][1,1][1,0]
| [0,0]*[1,0][1,0][1,1][1,1][0,0][0,0]*[1,0]
| [0,0]*[1,0][0,0][0,1][1,1][0,0][0,0][0,0]*[1,0]
| [0,0]*[1,0][1,0][0,0][0,1][0,0][0,0][0,0]*[1,0]
| [0,0]*[1,0][1,0][0,1][1,0][0,0][0,0][0,0]*[1,0]
\end{verbatim}
Then running Code~\ref{code:z-fact} in \verb|Walnut| returns \verb|TRUE|.
\qed
\end{proof}

\begin{theorem}
Let $c(\infw{rs})=(c_0,c_1,\ldots)$ be the $c$-factorization of the Rudin-Shapiro sequence $\infw{rs}$.
Then $c_m = \infw{rs}[i..i+n-1]$ where, for $0\le m \le 12$, $(i,n)$ belongs to 
    \[
    \{
    (0, 1), (1, 2), (3,1), (4, 5), (9, 3), (12, 2), (14, 5), (19, 5), (24, 4), (28, 8), (36, 12), (48, 6), (54, 6)
    \},
    \]
    and for all $m\ge 13$ with $p = \lfloor \frac{m}{4} \rfloor$,
 \[
    (i,n)
    = 
    \begin{cases}
		(13\cdot 2^{p-1} ,  2^{p}), &\text{if $m\equiv 0 \bmod{4}$;} \\
		(15\cdot 2^{p-1}  , 3\cdot 2^{p-1}), &\text{if $m\equiv 1 \bmod{4}$;} \\
		(9\cdot 2^{p}  ,  3\cdot2^{p}), &\text{if $m\equiv 2 \bmod{4}$;}\\
    (12\cdot 2^{p}  ,  2^{p}), &\text{if $m\equiv 3 \bmod{4}$.}
        \end{cases}
    \]
\end{theorem}
\begin{proof}
In Code~\ref{code:c-fact}, replace \verb|X| by \verb|RS|, \verb|k| by $2$, and \verb|LX| by
\begin{verbatim}
[0,0]*[0,1] | [0,0]*[0,1][1,0] | [0,0]*[1,0][1,1]  
| [0,0]*[1,1][0,0][0,1] | [0,0]*[1,0][0,0][0,1][1,1]
| [0,0]*[1,0][1,0][0,1][0,0] | [0,0]*[1,0][1,1][1,0][0,1]
| [0,0]*[1,0][0,0][0,1][1,0][1,1] | [0,0]*[1,0][1,0][0,1][0,0][0,0]
| [0,0]*[1,0][1,1][1,0][0,0][0,0]
| [0,0]*[1,0][0,0][0,1][1,1][0,0][0,0]
| [0,0]*[1,0][1,0][0,0][0,1][0,1][0,0]
| [0,0]*[1,0][1,0][0,0][1,1][1,1][0,0]
| [0,0]*[1,0][1,0][0,0][0,1][0,0][0,0][0,0]*[1,0]
| [0,0]*[1,0][1,0][0,1][1,0][0,0][0,0][0,0]*[1,0]
| [0,0]*[1,0][1,0][1,1][1,1][0,0][0,0][0,0]*[1,0]
| [0,0]*[1,0][0,0][0,1][1,1][0,0][0,0][0,0][0,0]*[1,0]
\end{verbatim}
Then running Code~\ref{code:c-fact} in \verb|Walnut| returns \verb|TRUE|.
\qed
\end{proof}

\subsection{The paper-folding sequence}

    The \emph{paper-folding} sequence arises from the iterative folding of a piece of paper. As the paper is folded repeatedly to the right and then unfolded, the sequence of turns is recorded. For each folding action, a corresponding binary digit is assigned: right turns are coded by $1$ and left turns by $-1$. This systematic recording process generates the infinite sequence
    \[
    \infw{pf}=11(-1)11(-1)(-1)111(-1)(-1)1(-1)(-1)111(-1)11\cdots.
    \]
    It is $2$-automatic and generated by the automaton in Figure~\ref{fig:DFAO:PF}.
    Finally, it can be written as $\nu(h^\omega(a))$, where $h\colon a \mapsto ab, b \mapsto cb, c \mapsto ad, d \mapsto cd$ and $\nu \colon a,b\mapsto 1, c,d\mapsto -1$.

\begin{theorem}
    Let $z(\infw{pf})=(z_0,z_1,\ldots)$ be the $z$-factorization of the paper-folding sequence $\infw{pf}$.
    Then, for all $m\in\{0,\ldots,5\}$, $z_m$ is given in Table~\ref{tab:z and c fact of our sequences} and, for all $m\ge 6$, $z_m=\infw{pf}[i..i+n-1]$ where
    \[
    (i,n)
    = 
     \begin{cases}
        (5\cdot 2^{m/2-1}, 2^{m/2-1}), &\text{if $m$ is even;} \\
        (3\cdot 2^{(m-1)/2},  2^{(m-1)/2+1}), &\text{if $m$ is odd.}
    \end{cases} 
    \]
\end{theorem}
\begin{proof}
In Code~\ref{code:z-fact}, replace \verb|X| by \verb|RS|, \verb|k| by $2$, and \verb|LX| by
\begin{verbatim}
[0,0]*[0,1] | [0,0]*[0,1][1,0] | [0,0]*[0,1][1,0][1,0] 
| [0,0]*[1,0][1,1][1,1] | [0,0]*[1,0][0,1][1,0][0,0]
| [0,0]*[1,0][1,1][1,1][0,0] | [0,0]*[1,0][0,0][1,1][0,0][0,0]*[0,0]
| [0,0]*[1,1][1,0][0,0][0,0][0,0]*[0,0]
\end{verbatim}
Then running Code~\ref{code:z-fact} in \verb|Walnut| returns \verb|TRUE|.
\qed
\end{proof}

\begin{theorem}
Let $c(\infw{pf})=(c_0,c_1,\ldots)$ be the $c$-factorization of the paper-folding sequence $\infw{pf}$.
Then $c_m = \infw{pf}[i..i+n-1]$ where, for $0\le m \le 9$, $(i,n)$ belongs to 
    \[
    \{
    (0, 1), (1, 1), (2,1), (3, 3), (6, 3), (9, 4), (13, 6), (19, 4), (23, 6), (29, 10)
    \},
    \]
    and for all $m\ge 10$ with $p = \lfloor \frac{m}{3} \rfloor$,
 \[
    (i,n)
    = 
    \begin{cases}
		(13\cdot 2^{p-2} - 1, 7\cdot 2^{p-2} ), &\text{if $m\equiv 0 \bmod{3}$;} \\
		(5\cdot 2^{p}-1,  2^{p} ), &\text{if $m\equiv 1 \bmod{3}$;} \\
		(3\cdot 2^{p+1} - 1,  2^{p-1}), &\text{if $m\equiv 2 \bmod{3}$.} 
	\end{cases}
    \]
\end{theorem}
\begin{proof}
In Code~\ref{code:c-fact}, replace \verb|X| by \verb|PF|, \verb|k| by $2$, and \verb|LX| by
\begin{verbatim}
[0,0]*[0,1] | [0,0]*[1,1] | [0,0]*[1,0][0,1] | [0,0]*[1,1][1,1]
| [0,0]*[1,0][1,1][0,1] | [0,0]*[1,0][0,1][0,0][1,0]
| [0,0]*[1,0][1,1][0,1][1,0] | [0,0]*[1,0][0,0][0,1][1,0][1,0]
| [0,0]*[1,0][0,0][1,1][1,1][1,0] | [0,0]*[1,0][1,1][1,0][0,1][1,0]
| [0,0]*[1,0][0,0][0,1][1,0][1,0][1,0][1,0]*
| [0,0]*[1,0][0,0][1,0][1,1][1,0][1,0][1,0]*
| [0,0]*[1,0][1,1][0,1][0,1][1,0][1,0][1,0]*
\end{verbatim}
Then running Code~\ref{code:c-fact} in \verb|Walnut| returns \verb|TRUE|.
\qed
\end{proof}

\subsection{The Mephisto-Waltz sequence}

The \emph{Mephisto-Waltz} sequence $\infw{mw} = aab aab bba \cdots$ is defined as the fixed
point of the morphism $a \mapsto aab$, $b \mapsto bba$ starting with $a$.
It is thus $3$-automatic and is generated by the automaton in Figure~\ref{fig:M:DFAO}. Another definition of this sequence is, for all $n\ge 0$, $\infw{mw}[n]=a$ if $|\rep_3(n)|_2$ is even, $\infw{mw}[n]=b$ otherwise, i.e., we store the parity of the number of $2$'s in the base-$3$ representation of $n$.

\begin{theorem}
\label{thm:z-fact-mw}
	Let $z(\infw{mw})=(z_0,z_1,\ldots)$ be the $z$-factorization of the Mephisto-Waltz sequence $\infw{mw}$.
	Then, for all $m\in\{0,\ldots,3\}$, $z_m$ is given in Table~\ref{tab:z and c fact of our sequences} and, for all $m\ge 4$, $z_m=\infw{mw}[i..i+n-1]$ where, for $p = \lfloor \frac{m}{3} \rfloor$,
	\[
	(i,n)
	= 
	\begin{cases}
		(8\cdot 3^{p-1} + 1, 2\cdot 3^{p-1}), &\text{if $m\equiv 0 \bmod{3}$;} \\
		(10\cdot 3^{p-1} +1 , 8\cdot 3^{p-1}), &\text{if $m\equiv 1 \bmod{3}$;} \\
		(2\cdot 3^{p+1} + 1 , 2\cdot 3^p), &\text{if $m\equiv 2 \bmod{3}$.} 
	\end{cases}
	\]
\end{theorem}
\begin{proof}
First, in \verb|Walnut|, code the Mephisto-Waltz sequence with the commands \verb|morphism h "0->001 1->110":| and \verb|promote MW h|.
Then, in Code~\ref{code:z-fact}, replace \verb|X| by \verb|MW|, \verb|k| by $3$, and \verb|LX| by
\begin{verbatim}
[0,0]*[0,1] | [0,0]*[1,2] | [0,0]*[1,1][0,1] | [0,0]*[2,1][1,1] 
| [0,0]*[1,0][0,2][2,2] | [0,0]*[2,0][2,2][0,0]*[1,0] 
| [0,0]*[1,0][0,2][1,2][0,0]*[1,0] | [0,0]*[2,0][0,2][0,0]*[1,0]":
\end{verbatim}
Then running Code~\ref{code:z-fact} in \verb|Walnut| returns \verb|TRUE|.
\qed
\end{proof}

\begin{theorem}
	Let $c(\infw{mw})=(c_0,c_1,\ldots)$ be the $c$-factorization of the Mephisto-Waltz sequence $\infw{mw}$.
	Then, for all $m\in\{0,\ldots,3\}$, $c_m$ is given in Table~\ref{tab:z and c fact of our sequences} and, for all $m\ge 4$, $c_m=\infw{mw}[i..i+n-1]$ where, for $p = \lfloor \frac{m}{3} \rfloor$,
	\[
	(i,n)
	= 
	\begin{cases}
		(10\cdot 3^{p-2} , 8\cdot 3^{p-2}), &\text{if $m\equiv 0 \bmod{3}$;} \\
		(2\cdot 3^{p}, 2\cdot 3^{p-1}), &\text{if $m\equiv 1 \bmod{3}$;} \\
		(8\cdot 3^{p-1}  , 2\cdot 3^{p-1}), &\text{if $m\equiv 2 \bmod{3}$.} 
	\end{cases}
	\]
\end{theorem}
\begin{proof}
In Code~\ref{code:c-fact}, replace \verb|X| by \verb|MW| coded in \verb|Walnut| as in the proof of Theorem~\ref{thm:z-fact-mw}, \verb|k| by $3$, and \verb|LX| by
\begin{verbatim}
[0,0]*[0,1] | [0,0]*[1,1] | [0,0]*[2,1] | [0,0]*[1,0][1,0] 
| [0,0]*[1,0][0,2][1,2][0,0]* | [0,0]*[2,0][0,2][0,0]* 
| [0,0]*[2,0][2,2][0,0]*":
\end{verbatim}
Then running Code~\ref{code:c-fact} in \verb|Walnut| returns \verb|TRUE|.
\qed
\end{proof}


\section{Conclusion}
\label{sec:conclusion}

In this paper, we investigated the following problem: given an abstract numeration system $S$ and an $S$-automatic sequence $\infw{x}$, is it possible to use \verb|Walnut| to obtain a description of both the Crochemore and Ziv-Lempel factorizations of $\infw{x}$ that only depend on the numeration system $S$?
We produced a detailed code for several classical automatic sequences in the Zeckendorff system as well as in bases $2$ and $3$.
According to us, a general answer to the previous question is far from being obvious to obtain.
Indeed, first, the software \verb|Walnut| only works in the case of so-called \emph{addable} abstract numeration systems, i.e., when addition can be performed by an automaton.
Then, as said previously, a candidate for the factorizations has to be known in advance in the hope of using \verb|Walnut|.
We believe that finding such candidates might be tricky for a general automatic sequence, when not much information is known about the inner structure of the sequence.
Observe also that, already among the $2$-automatic sequences we considered, the pairs of positions and lengths of the factors of the factorizations strongly depend on the sequence itself and not only on the underlying numeration system.
Finally, we wish to point that we examined non purely morphic sequences for which Berstel and Savelli write in~\cite[Sec.~6]{Berstel06} that ``it is not yet clear whether a satisfactory description [of the $c$-factorization] can be obtained''.


\subsection*{Acknowledgments}
We thank Narad Rampersad for useful discussions.

Manon Stipulanti is an FNRS Research Associate supported by the Research grant 1.C.104.24F.


%
%
%
 \bibliographystyle{splncs04}
 \bibliography{biblio.bib}

 \end{document}